%% file: Paper.tex
\documentclass[runningheads]{llncs}
\usepackage{amsmath,amssymb}
\usepackage{tikz}
\usepackage{hyperref}

\usepackage{subcaption} %For multiple figures next to each other

\usepackage[strings]{underscore}	 % otherwise DOI in bib do not work
\usepackage{petriGames1704} % Manunels PG style file for background / number is date of copy as continuously changed/extended
\usepackage{enumitem} % less space before itemize

\makeatletter
\newcommand{\xRightarrow}[2][]{\ext@arrow 0359\Rightarrowfill@{#1}{#2}}

\makeatother

\usetikzlibrary{arrows,petri,backgrounds, positioning, shapes, patterns,calc,automata}
\tikzstyle{envplace}=[circle,thick,draw=black!75,fill=white,minimum size=6mm]
\tikzstyle{sysplace}=[circle,thick,draw=black!75,fill=black!20,minimum size=6mm]
\tikzstyle{transition}=[rectangle,thick,draw=DarkBlue!75,fill=DarkBlue!20,minimum size=4mm]
\tikzstyle{bad}=[double=black!20]
\tikzstyle{badstate}=[pattern=checkerboard, fill opacity=0.1, draw opacity=1, draw=black, text opacity=1]
\tikzstyle{subgame}=[rectangle,thick,draw=DarkBlue!90,minimum size=25mm, line width = 0.5]
\tikzstyle{frame}=[rectangle, thick, draw=blue!50, minimum width = 120mm, minimum height=25mm, rounded corners=0.5cm, line width=0.8mm]
\tikzstyle{frame1}=[rectangle, thick, draw=orange!50, minimum width = 120mm, minimum height=40mm, rounded corners=0.5cm, line width=0.8mm]
\tikzstyle{draw1}=[draw=none]

\begin{document}

\title{Efficient Trace Encodings of Bounded Synthesis for Asynchronous Distributed Systems\thanks{This work was supported by the German Research Foundation (DFG) Grant Petri Games (392735815) and the Collaborative Research Center “Foundations of Perspicuous Software Systems” (TRR 248, 389792660), and by the European Research Council (ERC) Grant OSARES (683300).}}
\titlerunning{Efficient Trace Encodings for Asynchronous Distributed Systems}

\author{Jesko Hecking-Harbusch \and Niklas O.\ Metzger}
\authorrunning{J.\ Hecking-Harbusch \and N.O.\ Metzger}

\institute{
%\email{hecking-harbusch@react.uni-saarland.de}\\
%\email{s8nimetz@stud.uni-saarland.de}\\
Saarland University, Saarbr\"ucken, Germany
}

\maketitle

\input{InputFiles/abstract}

\section{Introduction}\label{sec:introduction}
\input{InputFiles/Introduction.tex}

\section{Example of Bounded Synthesis for Petri Games}\label{sec:motivation}
\input{InputFiles/PetriGames.tex}

\section{Background}\label{sec:background}

We introduce the necessary background on 
Petri nets~\cite{DBLP:books/sp/Reisig85a}, Petri games~\cite{petrigamessynthesis}, and the sequential encoding of bounded synthesis for Petri games~\cite{Finkbeiner15}.
Notice that we limit ourselves to \emph{$1$-bounded (safe)} Petri nets for simpler notation. 

\subsection{Petri Nets}

A ($1$-bounded) \textit{Petri net} $\petriNet$ consists of a set of \textit{places} $\pl$, a set of \textit{transitions} $\tr$, a \textit{flow relation} $\fl \subseteq (\pl \times \tr)\cup (\tr \times \pl)$, and an \textit{initial marking} $\init\subseteq\pl$.
The flow relation defines the \textit{arcs} from places to transitions~($\pl \times \tr$) and from transitions to places~($\tr \times \pl$). 
The state of a Petri net is represented by a \textit{marking} $M \subseteq \pl$ which positions one \textit{token} each in all places $p\in M$.
The elements of $\pl\cup\tr$ are considered as \textit{nodes}.
We define the \textit{preset} (and \textit{postset}) of a node~$x$ from Petri net~$\pNet$ as $\pre{\pNet}{x} = \{y \in \pl \cup \tr \mid (y,x)\in\fl\}$ (and $\post{\pNet}{x} = \{y \in \pl \cup \tr \mid (x,y)\in\fl\}$).
The preset and postset of transitions are non-empty and finite.
We use decorated names like $\pNet^\bd$ to also decorate the net's components. 
We abbreviate $\pre{\pNet^\bd}{x}$ and $\post{\pNet^\bd}{x}$ by $\pre{\bd}{x}$ and $\post{\bd}{x}$.
A transition~$t$ is \textit{enabled} at a marking $M$ if $\pre{\pNet}{t} \subseteq M$ holds (denoted by $M\firable{t}$).
An enabled transition~$t$ can be \textit{fired} from a marking~$M$ resulting in the successor marking $M' = (M \setminus \pre{\pNet}{t})\cup \post{\pNet}{t}$ (denoted by $M\firable{t}M'$).
We define the set of \textit{reachable markings} of a Petri net ${\reach(\pNet) = \{M \subseteq \pl \mid \exists t_1,..., t_n \in \tr : \exists M_1,...,M_n \subseteq \pl : \init\firable{t_1}M_1...\firable{t_n} M_n=M\}}$.
Two nodes $x, y$ are \textit{in conflict} (denoted by $\conflict{x}{y}$) if there exists a place $p\in\pl\setminus \{x,y\}$ from which $x$ and $y$ can be reached, exiting $p$ by different transitions. 

\subsection{(Bounded) Unfoldings and Subprocesses}

The \emph{unfolding} $\unf = (\pNet^U, \lambda^U)$ of a Petri net~$\pNet$ explicitly represents the \emph{causal pasts} of all places by eliminating all joins of places in the Petri net and separating these places into appropriate copies. 
Therefore, a loop in a Petri net results in an infinite unfolding.
The homomorphism  $\lambda^U : \pl^U \cup \tr^U \rightarrow \pl \cup \tr$ gives for nodes in the unfolding the corresponding original nodes.
For bounded synthesis, we consider \emph{bounded unfoldings}~$\unf^\bd = (\pNet^\bd, \lambda^\bd)$, where the memory bound $b$ defines how many causal pasts per place can be represented as separate copies.
Thereby, loops are only finitely often unfolded.
A net-theoretic \emph{subprocess} of a Petri net or an unfolding (denoted by $\sqsubseteq$) is defined by removing a set of transitions and all following places and transitions that cannot be reached anymore.

\subsection{Petri Games}\label{sec:petrigames}
\input{InputFiles/PetriGamesNew.tex}

\section{True Concurrency in Petri Games}\label{sec:trueCon}
\input{InputFiles/TrueConFiringSemantics.tex}

\section{Experimental Results}\label{sec:implementation}
\input{InputFiles/ImplementationAndResults}

\section{Related Work}\label{sec:relatedwork}

The \emph{control problem of asynchronous automata} is an alternative approach to the synthesis of distributed asynchronous systems with causal memory. 
The modeling with asynchronous automata does not allow the spawning and termination of players.
Also, it does not explicitly represent environment processes. 
Instead, every process can have uncontrollable behavior. 
The decidability of the control problem of asynchronous automata is open in general~\cite{DBLP:conf/icalp/Muscholl15}. 
There are some decidability results for the control problem of asynchronous automata for restrictions on the dependencies of actions~\cite{DBLP:conf/fsttcs/GastinLZ04} or on the synchronization behavior~\cite{DBLP:conf/concur/MadhusudanT02,DBLP:conf/fsttcs/MadhusudanTY05}. 
Decidability has also been obtained for acyclic communication structures~\cite{DBLP:conf/icalp/GenestGMW13,DBLP:conf/fsttcs/MuschollW14}. 
The class of \emph{Decomposable games}~\cite{DBLP:conf/fsttcs/Gimbert17} proposes a new proof technique to unify and extend these results.
Recently, an exponential gap between the control problem of asynchronous automata and Petri games has been identified~\cite{CONCUR19}.

There is a broad theory and several implementations for model checking of distributed systems:
For Petri nets as representation of distributed systems, it often suffices to only consider finite prefixes of the unfolding~\cite{DBLP:journals/acta/KhomenkoKV03,DBLP:series/eatcs/EsparzaH08,DBLP:journals/tcs/BonetHKTV14}.
It is most interesting whether these results can be lifted to Petri games and causal past. 
Partial order reduction and true concurrency have been studied thoroughly to speed up the model checking of finite distributed systems~\cite{DBLP:journals/fuin/Heljanko99,DBLP:phd/basesearch/Heljanko02,DBLP:conf/popl/FlanaganG05,DBLP:conf/wotug/MeulenP09}.
The systems we synthesize are especially powerful as both the system and the environment can run infinitely and non-determinism of the environment is represented.   

\section{Conclusion}\label{sec:Conclusion}

We presented how to utilize concurrency in bounded synthesis for asynchronous distributed systems by firing as many true concurrent transitions as possible in our new true concurrent encoding.
The previous sequential encoding enumerated all interleavings.
For the true concurrent encoding, we represent the decisions of the environment players explicitly as environment strategies for Petri games and showed that this enables us to fire all enabled transitions as early as possible while maintaining the existence of winning system strategies.
The experimental results show that our tool implementation of the true concurrent encoding outperforms the sequential encoding on all benchmark families by a considerable margin.
Even in the rare case of benchmark families without true concurrent transitions, the true concurrent encoding slightly outperforms the sequential encoding despite resulting in larger QBFs.

For future work, we want to apply environment strategies and true concurrency in \textsc{Adam} to improve synthesis for a bounded number of system players and one environment players.
Furthermore, we plan to extend the bounded synthesis encoding further.
On the one hand, we want to identify disconnected parts of the Petri game, solve them in isolation, and compose them back together.
On the other hand, we plan to extend the expressivity of considered winning conditions. 
Local liveness conditions of places to reach should be straightforward whereas global winning conditions in the form of markings to reach or avoid could prove difficult for the true concurrent encoding as certain interleavings may be skipped.
Therefore, we believe that local winning conditions on the progress of individual tokens could be a good middle ground between the current local winning conditions of bad places and global winning conditions.

\bibliographystyle{splncs04} 
\bibliography{bibfile}

\end{document}

%% file: InputFiles/abstract.tex
% !TeX root = ../Paper.tex
\begin{abstract}
	The manual implementation of distributed systems is an error-prone task because of the asynchronous interplay of components and the environment. 
	Bounded synthesis automatically generates an implementation for the specification of the distributed system if one exists. 
	So far, bounded synthesis for distributed systems does not utilize their asynchronous nature. 
	Instead, concurrent behavior of components is encoded by all interleavings and only then checked against the specification. 
	We close this gap by identifying true concurrency in synthesis of asynchronous distributed systems represented as Petri games. 
	This defines when several interleavings can be subsumed by one true concurrent trace. 
	Thereby, fewer and shorter verification problems have to be solved in each iteration of the bounded synthesis algorithm. 
	For Petri games, experimental results show that our implementation using true concurrency outperforms the implementation based on checking all interleavings. 
\end{abstract}

%% file: InputFiles/Introduction.tex
One ambitious goal in computer science is the automatic generation of programs. 
For a given specification, a \emph{synthesis algorithm} either generates a program satisfying the specification or determines that no such program exists. 
Nowadays, most synthesis tools deploy a game-theoretic view on the problem~\cite{DBLP:conf/cav/JobstmannGWB07,DBLP:conf/tacas/Ehlers11,DBLP:conf/cav/BohyBFJR12,DBLP:conf/tacas/FaymonvilleFRT17}. 
The synthesis of \emph{distributed systems}~\cite{DBLP:conf/icalp/PnueliR89} can be represented by a team of system players and a team of environment players playing against each other. 
Each system player acts on individual information and requires a local strategy, which in combination with the strategies of the other system players satisfies an objective against the decisions of the team of environment players. 
The environment players can cooperate to prevent the satisfaction of the objective by the system players. 
In the \emph{synchronous} setting where all players progress at the same rate, the synthesis problem for distributed systems is undecidable~\cite{DBLP:conf/focs/PnueliR90,DBLP:conf/lics/FinkbeinerS05}.

\emph{Petri games} represent \emph{asynchronous} behavior in the synthesis of distributed systems where processes can progress at individual rates between synchronizations. 
Furthermore, the players of the team of system players have \emph{causal memory}, i.e., a system player can base decisions on its local past and the local past of all other players up to their last synchronization. 
The synthesis problem for Petri games is decidable if for a maximum of one for the number of system players or the number of environment players~\cite{petrigamessynthesis,DBLP:conf/fsttcs/FinkbeinerG17}. 
If the restrictions on the team size cannot be met, \emph{bounded synthesis}~\cite{boundedsynthesis} is applied to incrementally increase the memory of possible system strategies until a winning one is found.

Each iteration of the bounded synthesis algorithm for Petri games~\cite{Finkbeiner15} checks the existence of a winning system strategy with bounded memory by simulating the resulting Petri game. 
This simulation is represented as all \emph{interleavings} of fired transitions allowed by possible system strategies.
For two independent decisions, it makes no difference whether one decision or the other is scheduled first.
It suffices to only check one scheduling where both decisions happen \emph{true concurrently}.
The true concurrent scheduling not only considers fewer schedulings but also shorter ones.
Furthermore, the true concurrent scheduling enables us to refine the detection of loops in bounded synthesis.
This results in a considerable speed-up of the verification part of bounded synthesis for Petri games. 

To identify true concurrency, we introduce \emph{environment strategies} for Petri games which explicitly represent the decisions of environment players.
Environment strategies restrict a given system strategy and try to reach markings which prove the system strategy to \emph{not} be winning.
We present how the explicit environment decisions of environment strategies allow the firing of maximal sets of true concurrent transitions while preserving the applicability to bounded synthesis.
This requires some \emph{stalling} options for the environment.
For bounded synthesis, we encode the assumptions on system and environment strategies as well as the winning objective of Petri games as \emph{quantified Boolean formula} (QBF). 
We compare the implementations of the \emph{sequential encoding} based on all interleavings and our new \emph{true concurrent encoding} on an extended set of benchmarks\footnote{The sequential and the concurrent encoding can be tested online as part of the \adam\ toolkit~\cite{ADAM}: \url{https://www.react.uni-saarland.de/tools/online/ADAM/}}.
Our experimental results show that the true concurrent encoding outperforms the sequential encoding by a considerable margin. 

The key contributions of this paper are the following:
\begin{itemize}
	\item We develop the theoretical foundation of \emph{true concurrency} of components in synthesis for asynchronous distributed systems by representing environment decisions explicitly in \emph{environment strategies} of Petri games.
	\item We prove that environment strategies \emph{preserve existence of winning system strategies} and encode them as QBFs for bounded synthesis for Petri games.
	\item We \emph{implement} the true concurrent encoding and show considerable improvements against the sequential encoding on an extended benchmark set.
\end{itemize}

The paper is structured as follows: In Section~\ref{sec:motivation}, we give an intuitive introduction to Petri games and the benefits of true concurrent scheduling for bounded synthesis for Petri games. 
Section~\ref{sec:background} recalls the required background on Petri nets, Petri games, and bounded synthesis. 
In Section~\ref{sec:trueCon}, we introduce environment strategies and prove that they preserve the existence of winning strategies.
Section~\ref{sec:qbf} gives the true concurrent encoding formally as QBF. 
Section~\ref{sec:implementation} surveys experimental results for the implementation of the true concurrent encoding.

%% file: InputFiles/PetriGames.tex
% !TeX root = ../Paper.tex

\input{InputFigures/ProductionLine}

Figure \ref{fig:ProductionLine} illustrates how Petri games represent the synthesis problem of asynchronous distributed systems and how true concurrency simplifies bounded synthesis for Petri games. 
This Petri game specifies a production line for repairing a product.
The different possible requirements for repair are modeled as choices of the environment. 
The product can either require repair by a single robot or by both robots concurrently.
These robots are represented by system players and have to collectively meet the requirement of the product.

Petri games are based on an underlying Petri net and distribute the places into two disjoint sets for the team of environment players and for the team of system players. 
White places belong to the environment and represent the product and its requirements for repair.
Gray places belong to the system and represent the robots of the production line. 
The players are represented as tokens and their team is determined by the type of the place they are residing in. 
Initially, there is one token in the place \textit{env} representing an environment player.
Transitions define the flow of tokens through the Petri game as in Petri nets. 
When all places before a transition contain a token, then this transition is enabled. 
Firing an enabled transition consumes the tokens in all places before the transition and produces tokens in all places after it. 
The firing of enabled transition \textit{1\_robot} results in a consumption of the token in \textit{env} and the production of tokens in places \textit{1\_robot\_check}, \textit{env1}, \textit{robot1}, \textit{env2}, and \textit{robot2}.
By this transition, both robots are started and it is required that only the first one repairs a part of the product.

The winning objective of the game is represented by the bad place $\bot$.
The team of system players has to avoid reaching this place for all choices of the environment.
Based on its causal past, a system player can decide which outgoing transitions to fire. 
For example, the system place \textit{robot2} can either be reached via transition \textit{1\_robot} or via \textit{2\_robots} and then the player can decide in both cases independently between transitions \textit{repair2} and \textit{ignore2}.
Deciding independently is necessary because if the environment has chosen \textit{1\_robot}, no repair by the second robot is allowed whereas if the environment has chosen \textit{2\_robots}, repair by the second robot is required.
The winning system strategy is presented in Fig.~\ref{fig:ProductionLineStrategy} where primed places and transitions result from different causal pasts.
The outgoing transitions \textit{ignore2} of place \textit{robot2} and \textit{repair2'} of \textit{robot2'} represent the necessary different decisions of the system.
Notice that the bad place is not reachable based on the decisions in the winning system strategy.

\input{InputFigures/ProductionLineStrategy}

Bounded synthesis for Petri games uses quantified Boolean formulas (QBFs) to decide the existence of a winning system strategy for a given memory bound. 
The decisions at system places are represented explicitly as existentially quantified variables which are tested to be avoiding bad places for subsequent distributions of tokens until the game either terminates or reaches a loop.
The memory bound implies the length of these sequences.
The sequential encoding tests all possible interleavings of transitions, e.g., in our example, first the environment makes a decision between \textit{1\_robot} and \textit{2\_robots} and then two interleavings are tested depending on the ordering of the decisions of both system players.  
Our new concurrent flow semantics identifies such situations and replaces them with one true concurrent step for the decisions of both robots. 
Thereby, we reduce the number of considered traces from four interleavings of length three to two true concurrent traces of length two to verify the winning system strategy of Fig.~\ref{fig:ProductionLineStrategy}.

%% file: InputFigures/ProductionLine.tex
% !TeX root = ../Paper.tex
	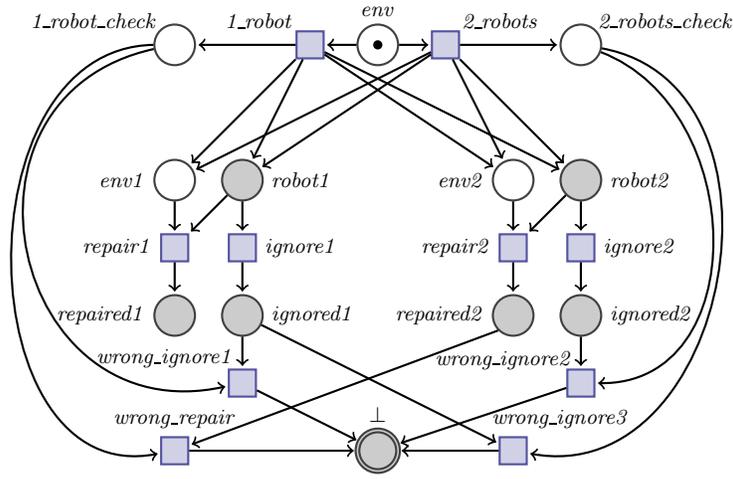
\begin{figure}[t]
	%\captionsetup[subfigure]{font=footnotesize}
	\centering

	\begin{tikzpicture}[scale=0.9,every node/.style={transform shape}]
	%Gobal Flow places
	\node [envplace] (env) [tokens = 1, label=above:\textit{env}]{};
	%\node [sysplace] (sys1) [tokens = 1, left of = env, label = above:\textit{robot1}] {};
	%\node [sysplace] (sys2) [tokens = 1, right of = env, label = above:\textit{robot2}] {};
	
	%\node [sysplace] (sysstart) [below of = env, below of = env,label=below right:\textit{sysstart}]{};
	
	\node [envplace] (e1) [left of = env, left of=env, left of = env, below of = env, below of = env, label=left:\textit{env1}]{};
	\node [envplace] (e2) [right of=env,  right of = env, below of = env, below of = env, label=left:\textit{env2}]{};
	\node [sysplace] (s1) [ right of = e1, label=right:\textit{robot1}]{};
	\node [sysplace] (s2) [right of = e2, label=right:\textit{robot2}]{};
	
	\node [sysplace] (done11) [below of = s1, below of = s1, label =  right:\textit{ignored1}]{};
	\node [sysplace] (done12) [below of = s1, below of = s1,  left of = s1, label = left:\textit{repaired1}]{};
	\node [sysplace] (done21) [below of = s2, below of = s2,  label = right:\textit{ignored2}]{};
	\node [sysplace] (done22) [below of = e2, below of = e2,  label = left:\textit{repaired2}]{};
	
	%Bad behavior check places
	\node [envplace] (tcheck1) [left of= env, left of= env, left of = env, label={[label distance =-0.12cm]above left:\textit{1\_robot\_check}}]{};

	\node [envplace] (tcheck2) [right of= env, right of= env, right of = env, label={[label distance=-0.12cm]above right:\textit{2\_robots\_check}}]{};

	%Syscheck for right decision of systems:
	
	\node [sysplace, bad] (bad) [below of = done11, below of = done11, right of = done11,  right of = done11, label = above:$\bot$]{};
	
	%Flow transitions
	%\node [transition] (sysinit) [below of = sysstart,  label=below :\textit{init}]{};
	\node [transition](t1) [left of= env, label={[label distance=-0.12cm]above left:\textit{1\_robot}}]{};
	\node [transition](t2) [right of=env, label={[label distance=-0.12cm]above right:\textit{2\_robots}}]{};

	\node [transition] (tdec11) [below of = s1, label =  right:\textit{ignore1}]{};
	\node [transition] (tdec12) [below of = s1, left of= s1, label = left:\textit{repair1}]{};
	
	\node [transition] (tdec21) [below of = s2, label =  right:\textit{ignore2}]{};
	\node [transition] (tdec22) [below of = e2, label = left:\textit{repair2}]{};

	%Bad behavior transitions
	\node [transition] (tbad3) [ below of= done11, label = {[label distance=-0.1cm]95:\textit{wrong\_ignore1}}]{};
	\node [transition] (tbad4) [ below of= done21, label ={[label distance=-0.12cm]95:\textit{wrong\_ignore2}}]{};
	\node [transition] (tbad5) [ below of = done12, below of= done12, label = above :\textit{wrong\_repair}]{};
	\node [transition] (tbad6) [ below of = done22, below of= done22, label = {[xshift=0.7cm]above:\textit{wrong\_ignore3}}]{};

	%Flow paths
	%\path[->, thick] (sysinit) 
	%edge [pre]  (sysstart)
	%edge [post, bend right] (s1)
	%edge [post, bend left] (s2);
		
	\path[->, thick] (t1) 
	%edge [pre]  (sys1)
	%edge [pre]  (sys2)
	edge [pre]  (env)
	edge [post] (tcheck1)
	edge [post] (s2)
	edge [post] (s1)
	edge [post] (e2)
	edge [post] (e1);
	
	\path[->, thick] (t2) 
	%edge [pre]  (sys1)
	%edge [pre]  (sys2)
	edge [pre]  (env)
	edge [post] (tcheck2)
	edge [post] (s2)
	edge [post] (s1)
	edge [post] (e1)
	edge [post] (e2);
	
	\path[->, thick] (tdec11) 
	edge [pre] (s1)
	edge [post] (done11);
	
	\path[->, thick] (tdec12) 
	edge [pre]  (e1)
	edge [pre] (s1)
	edge [post] (done12);
	
	\path[->, thick] (tdec21) 
	edge [pre] (s2)
	edge [post] (done21);
	
	\path[->, thick] (tdec22) 
	edge [pre]  (e2)
	edge [pre] (s2)
	edge [post] (done22);

	\path [->, thick] (tbad3)
	edge [pre] (done11)
	edge [pre, looseness = 1.65, bend left, in = 90, out = 90] (tcheck1)
	edge [post] (bad);
	
	\path [->, thick] (tbad4)
	edge [pre] (done21)
	edge [pre, looseness = 1.22, bend right, in = 250, out = 270] (tcheck2)
	edge [post] (bad);

	\path [->, thick] (tbad5)
	edge [pre] (done22.225)
	edge [pre, looseness = 1.22, bend left, in = 90, out = 100] (tcheck1)
	edge [post] (bad);
		
	\path [->, thick] (tbad6)
	edge [pre] (done11)
	edge [pre, looseness = 1.28, bend right, in = 270, out = 270] (tcheck2)
	edge [post] (bad);
	\end{tikzpicture}
	\caption{This Petri game specifies a production line where two robots can repair a product. The product either requires repair by only one or by both robots.
	}
	\label{fig:ProductionLine}
\end{figure}

%% file: InputFigures/ProductionLineStrategy.tex
% !TeX root = ../Paper.tex
	\begin{figure}[t]
	%\captionsetup[subfigure]{font=footnotesize}
	\centering
	
	\begin{tikzpicture}[scale=0.9,every node/.style={transform shape}]
	%Gobal Flow places
	\node [envplace] (env) [tokens = 1, label=above:\textit{env}]{};
	%\node [sysplace] (sys1) [tokens = 1, left of = env, label = above:\textit{robot1}] {};
	%\node [sysplace] (sys2) [tokens = 1, right of = env, label = above:\textit{robot2}] {};
	
	%\node [sysplace] (sysstart) [below of = env, below of = env,label=below right:\textit{sysstart}]{};
	
	\node [envplace] (e1) [left of = env, left of = env,  left of=env, left of = env,   below of = env, label=left:\textit{env1}]{};
	\node [envplace] (e2') [right of=env,   right of = env, right of = env,  below of = env, label= {below:\textit{env2'}}]{};
	\node [sysplace] (s1) [right of = e1, label=below:\textit{robot1}]{};
	\node [envplace] (e2) [right of = s1,  label=below:\textit{env2}]{};
    \node [sysplace] (s2) [right of = e2,  label={[label distance=-0.15cm]below right:\textit{robot2}}]{};
	\node [sysplace] (s2') [right of = e2', label=right:\textit{robot2'}]{};
	\node [sysplace] (s1') [left of = s2', left of = s2', label=below:\textit{robot1'}]{};
	\node [envplace] (e1') [left of = s1', label={[label distance=-0.15cm]above left:\textit{env1'}}]{};
	
	%\node [sysplace] (done11) [below of = s1, below of = s1, label = above:\textit{ignored1}]{};
	\node [sysplace] (done12) [below of = s1, below of = s1,  left of = s1, label =left:\textit{repaired1}]{};
	\node [sysplace] (done21) [below of = s2, below of = s2,  label = {[label distance=-0.05cm]200:\textit{ignored2}}]{};
	%\node [sysplace] (done22) [below of = s2, below of = s2,  right of = s2, label =above:\textit{repaired2}]{};
	
	%\node [sysplace] (done11') [below of = s1', below of = s1', label = above:\textit{ignored1'}]{};
	\node [sysplace] (done12') [below of = s1', below of = s1',  left of = s1', label = {[label distance=-0.12cm]320:\textit{repaired1'}}]{};
	%\node [sysplace] (done21') [below of = s2', below of = s2',  label = below:\textit{ignored2'}]{};
	\node [sysplace] (done22') [below of = s2', below of = s2',   label = right:\textit{repaired2'}]{};
	
	%Bad behavior check places
	\node [envplace] (tcheck1) [left of= env, left of= env, left of = env, label=left:\textit{1\_robot\_check}]{};

	\node [envplace] (tcheck2) [right of= env, right of= env, right of = env, label=right:\textit{2\_robots\_check}]{};

	%Syscheck for right decision of systems:

	%Flow transitions
	%\node [transition] (sysinit) [below of = sysstart,  label=below :\textit{init}]{};
	\node [transition](t1) [left of= env, label={[label distance = -0.12cm]above left:\textit{1\_robot}}]{};
	\node [transition](t2) [right of=env, label={[label distance= -0.12cm]above right:\textit{2\_robots}}]{};

	%\node [transition] (tdec11) [below of = s1, label =  right:\textit{ignore1}]{};
	\node [transition] (tdec12) [below of = s1, left of= s1, label = left:\textit{repair1}]{};
	
	\node [transition] (tdec21) [below of = s2, label =  {[label distance=-0.05cm]200:\textit{ignore2}}]{};
	%\node [transition] (tdec22) [below of = s2, right of= s2, label = right:\textit{repair2}]{};
	
	%\node [transition] (tdec11') [below of = s1', label =  right:\textit{ignore1}]{};
	\node [transition] (tdec12') [below of = s1', left of= s1', label = {[label distance=-0.05cm]350:\textit{repair1'}}]{};
	
	%\node [transition] (tdec21') [below of = s2', label =  left:\textit{ignore2}]{};
	\node [transition] (tdec22') [below of = s2', label = right:\textit{repair2'}]{};

	\path[->, thick] (t1) 
	%edge [pre]  (sys1)
	%edge [pre]  (sys2)
	edge [pre]  (env)
	edge [post] (tcheck1)
	edge [post] (s2)
	edge [post] (s1)
	edge [post] (e2)
	edge [post] (e1.45);
	
	\path[->, thick] (t2) 
	%edge [pre]  (sys1)
	%edge [pre]  (sys2)
	edge [pre]  (env)
	edge [post] (tcheck2)
	edge [post] (s2'.135)
	edge [post] (s1')
	edge [post] (e1')
	edge [post] (e2');
	
	\path[->, thick] (tdec12) 
	edge [pre]  (e1)
	edge [pre] (s1)
	edge [post] (done12);
	
	\path[->, thick] (tdec21) 
	edge [pre] (s2)
	edge [post] (done21);
	
	\path[->, thick] (tdec12') 
	edge [pre]  (e1')
	edge [pre] (s1')
	edge [post] (done12');
	
	\path[->, thick] (tdec22') 
	edge [pre]  (e2')
	edge [pre] (s2')
	edge [post] (done22');

	\end{tikzpicture}
	\caption{A winning system strategy is presented for the Petri game from \refFig{ProductionLine}, which specifies a production line with two robots. 
	The system players make different decisions depending on the choice of the environment. 
	Transitions which cannot be enabled and unreachable places are removed.
	}
	\label{fig:ProductionLineStrategy}
\end{figure}
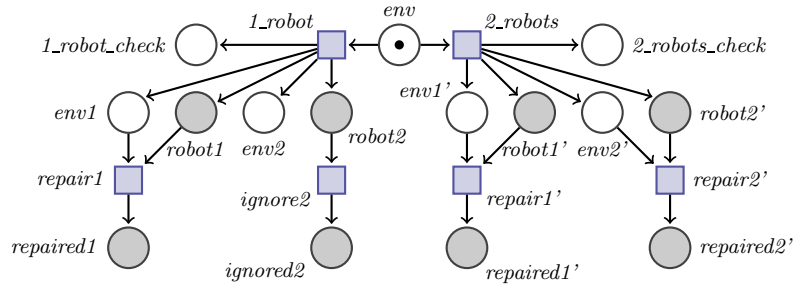

%% file: InputFiles/PetriGamesNew.tex
A \textit{Petri game} $\petriGameRecall$~\cite{petrigamessynthesis} with $\bad\subseteq\plS\cup\plE$ has an underlying Petri net $\petriNet$ with $\pl = \plS \uplus \plE$. 
The sets $\plS$, $\plE$, and~$\bad$ define the \emph{system places}, the \emph{environment places}, and the \emph{bad places}.
Unfoldings translate from Petri nets to Petri games by keeping the classification of places as system, environment, and bad places.
A \emph{system strategy} for $\pGame$ is a subprocess $\sysstrat = (\pNet^\sysstrat, \lambda^\sysstrat)$ of the unfolding $\unf= (\pNet^U,\lambda^U)$ of $\pGame$ where system places can remove outgoing transitions such that the following requirements hold:
	\begin{itemize}[topsep=5pt]
		\item[(S1)] \emph{Determinism}:\\ $\forall M \in \reach(\pNet^\sysstrat) : \forall p \in M \cap \plS^\sysstrat : \exists^{\leq 1} t \in \tr^\sysstrat : p \in \pre{\sysstrat}{t} \wedge \pre{\sysstrat}{t} \subseteq M$
		\item[(S2)] \emph{System refusal}: $\forall t \in \tr^U : t \notin \tr^\sysstrat \wedge \pre{\sysstrat}{t} \subseteq \pl^\sysstrat \implies ( \exists p \in \pre{\sysstrat}{t}\cap \plS^\sysstrat : \forall t' \in \post{U}{p} : \lambda^U(t) = \lambda^U(t') \implies t' \notin \tr^\sysstrat )$
		\item[(S3)] \emph{Deadlock-avoidance}:\\ $\forall M \in \reach(\pNet^\sysstrat) : \exists t_U \in \tr^U : \pre{U}{t_U} \subseteq M \implies \exists t_\sysstrat \in \tr^\sysstrat : \pre{\sysstrat}{t_\sysstrat} \subseteq M$
	\end{itemize}

\emph{Determinism} requires each system player to have at most one transition enabled for all reachable markings. 
\emph{System refusal} requires that the removal of a transition from the unfolding is based on a system place deleting all outgoing copies of that transition.
This enforces that system players base their decisions only on their causal past. 
\emph{Deadlock-avoidance} requires the system strategy to enable at least one transition for each reachable marking as long as one transition is enabled in the unfolding. 
A system strategy is \emph{winning} for the winning condition \emph{safety} if no bad place can be reached in the system strategy, i.e., $\forall M \in \reach(\pNet^\sysstrat)  : \lambda^\sysstrat[M] \cap \bad = \emptyset$.
The synthesis problem for Petri games with safety as winning objective is EXPTIME-complete if we limit the number of system players or the number of environment players to one~\cite{petrigamessynthesis,DBLP:conf/fsttcs/FinkbeinerG17}. 

\subsection{Sequential Encoding of Bounded Synthesis for Petri Games}

The bounded synthesis algorithm~\cite{Finkbeiner15} takes a Petri game and increases the memory bound $b$ until a winning system strategy is found (or runs forever).
The finite bounded unfolding~$\unf^\bd = (\pNet^\bd, \lambda^\bd)$ is used to encode the existence of a winning system strategy (as variables~$\mathscr{S}^b$) for all sequences of markings (as variables~$\mathscr{M}_n$) up to the maximal simulation length~$n \leq 2^{|\pl^\bd|} + 1$ as QBF.
In the encoding, concurrent transitions are represented by all possible interleavings as between two markings only a single transition is fired. 
For readability, we abbreviate $\pre{\pNet^b}{x}$ by $^\bullet x$ and $\post{\pNet^b}{x}$ by $x^\bullet$.
The QBF has the form $\exists \mathscr{S}^b : \forall \mathscr{M}_n : \phi_n$ where 
		$\mathscr{S}^b \defEQ \{ (p,\lambda^b(t)) \mid p \in\pl^b_S \land t\in p^\bullet \}$ and $\mathscr{M}_n \defEQ \{ (p,i) \mid p \in \pl^b \land 1\leq i \leq n \}$.

The system strategy~$\mathscr{S}^b$ consists of Boolean variables representing the system's choice for each pair of system place in the bounded unfolding and outgoing transition of the corresponding system place in the original game.
This encoding ensures that each system strategy satisfies \emph{system refusal} (S2) because neither pure environment transitions can be disabled nor can transitions be differentiated due to the bounded unfolding.  
The marking sequence $\mathscr{M}_n$ contains Boolean variables for each pair of place in the bounded unfolding and number $1 \leq i \leq n$ to encode in which of the $n$ subsequent markings this place is contained.  

The matrix $\phi_n$ of the QBF $\exists \mathscr{S}^b : \forall \mathscr{M}_n : \phi_n$ is defined as follows:
\begingroup
\allowdisplaybreaks
\begin{eqnarray*}
	\phi_n & \defEQ &  \bigwedge_{1\leq i < n} \Bigg( \mathit{sequence}_i \implies \mathit{win}_i \Bigg) \wedge \big( \mathit{sequence}_n \implies \mathit{win}_n \wedge \mathit{loop} \big)\\
	\mathit{sequence}_i & \defEQ & \mathit{initial} \wedge \mathit{seqflow}_1 \wedge \mathit{seqflow}_2 \wedge \dots \wedge \mathit{seqflow}_{i-1}\\
	\mathit{initial} & \defEQ &  \bigwedge_{p\in \init^b} (p,1) \wedge \bigwedge_{p\in \pl^b \setminus \init^b} \neg(p,1)\\
	\mathit{seqflow}_i & \defEQ &  \bigvee_{t\in\tr^b} \Bigg( \bigwedge_{p\in^\bullet t}(p,i) \wedge \bigwedge_{p \in^\bullet t \cap \plS^b} (p,\lambda^\bd(t)) \wedge  \bigwedge_{p\in t^\bullet} (p,i+1) \wedge\\
	&&~~~~~ \bigwedge_{p \in ^\bullet t\setminus t^\bullet} \neg (p,i+1)  \wedge \bigwedge_{p \in \pl^b \setminus (^\bullet t \cup t^\bullet)} \big((p,i) \iff (p,i+1) \big) \Bigg)
\end{eqnarray*}	
\endgroup

For each simulation point $1\leq i \leq n$, it is tested whether the variables in~$\mathscr{M}_n$ represent a correct $\mathit{sequence}_i$ of markings up to $i$ corresponding to a play in the bounded unfolding. 
If this is the case then \(\mathit{win}_i\) tests whether the marking at $i$ fulfills the requirements to be winning. 
If $i=n$, i.e., the limit on the simulation is reached, it is additionally tested that a $\mathit{loop}$ occurred. 
A correct \emph{sequence} of markings starts from the \emph{initial} marking followed by the \emph{sequential flow} of $i-1$ enabled and by the system strategy allowed transitions.
The \emph{sequential flow} of a transition from time point $i$ requires all places in its preset to contain a token and the system strategy of system places in its preset to allow the transition.
Then, at $i+1$, the places of its postset are set to true, places in its preset but not its postset are set to false, and all other places retain their truth value.
\begingroup
\allowdisplaybreaks
\begin{eqnarray*}
	\mathit{win}_i & \defEQ & \mathit{nobadplace}_i \wedge \mathit{deterministic}_i \wedge \big(\mathit{deadlock}_i \implies \mathit{terminating}_i\big) \\
	\mathit{nobadplace}_i & \defEQ & \bigwedge_{p\in\bad^b} \neg(p,i) \\
	\mathit{deterministic}_i & \defEQ & \bigwedge_{\substack{t_1,t_2\in\tr,t_1\neq t_2, \\ {^\bullet t_1} \cap {^\bullet t_2} \cap \plS^b \neq \emptyset}} \Bigg( \bigvee_{p\in^\bullet t_1 \cup ^\bullet t_2}\hspace{-0.5cm}\neg(p,i) \vee\hspace{-0.5cm}\bigvee_{\substack{p_1\in ^\bullet t_1 \cap \plS^b, \\ p_2\in ^\bullet t_2 \cap \plS^b}}\hspace{-0.5cm}\neg (p_1,\lambda^\bd(t_1)) \vee \neg (p_2,\lambda^\bd(t_2)) \Bigg)\\
	\mathit{deadlock}_i & \defEQ & \bigwedge_{t\in\tr^b} \Bigg( \bigvee_{p\in^\bullet t} \neg(p,i) \vee \bigvee_{p\in ^\bullet t \cap \plS^b} \neg (p,\lambda^\bd(t)) \Bigg)\\
	\mathit{terminating}_i & \defEQ & \bigwedge_{t \in \tr^b} \Bigg( \bigvee_{p\in ^\bullet t} \neg(p,i) \Bigg)\\
	\mathit{loop} & \defEQ & \bigvee_{1 \leq i_1 < i_2 \leq n} \Bigg(\bigwedge_{p\in\pl^b} \big((p,i_1) \iff (p,i_2) \big) \Bigg)
\end{eqnarray*}
\endgroup

If \(\mathit{sequence}_i\) is fulfilled then \(\mathit{win}_i\) tests whether the last marking fulfills the requirements to be winning at $i$. 
If $i=n$, i.e., the limit on the simulation is reached, it is additionally tested that a $\mathit{loop}$ occurred. 
The play is winning if \emph{no bad place} is reached, the system makes only $\mathit{deterministic}$ decisions (S1), and each $\mathit{deadlock}$ is caused by $\mathit{termination}$ (S3). 
A deadlock occurs when no transition is enabled including the choices of the system strategy~$\mathscr{S}^b$. 
Meanwhile, termination occurs when no transition is enabled independently of the system strategy. 
Therefore, $\mathit{deadlock}_i \implies \mathit{terminating}_i$ ensures that the system does not prevent the reaching of bad places by stopping to fire transitions, but deadlocks are only allowed when the entire game terminates. 
A \emph{loop} in a Petri game occurs when the same marking is repeated at two different simulation points. 
As the system strategy has to be deterministic, its behavior repeats infinitely often in the loop such that the system strategy is also winning in an infinite play.

%% file: InputFiles/TrueConFiringSemantics.tex
% !TeX root = ../Paper.tex
In this section, we define true concurrency in Petri games.
Therefore, we first formalize \emph{environment strategies} to explicitly represent environment decisions in response to a given system strategy.  
This enables us to define the \emph{true concurrent flow semantics} for Petri games, which enforces that transitions are fired as early and as parallel as possible. 
We prove that this semantics agrees with the interleaving semantics on the existence of a winning strategy for the system.

\subsection{Environment Strategy}

System strategies represent the system's restrictions of enabled transitions but purely environmental transitions remain uncontrollable.
Therefore, a system strategy can result in different fired transitions. 
We introduce \emph{environment strategies} to explicitly represent decisions of environment players and to obtain a unique sequence of fired transitions up to reordering of independent transitions. 

An \emph{environment strategy} $\envstrat = (\pNet^\envstrat, \lambda^\envstrat)$ is a subprocess of a system strategy $\sysstrat = (\pNet^\sysstrat, \lambda^\sysstrat)$ (which, in turn, is a subprocess of the unfolding $\unf = (\pNet^{U},\lambda^U)$ of the given Petri game~$\pGame$) where environment places can remove outgoing transitions such that the following three requirements hold:
\begin{itemize}[topsep=5pt]
	\item[(E1)] \emph{Explicit choice}: $ \forall p \in \pl^\envstrat_E : \exists^{\leq 1} t \in \tr^\envstrat : p \in \pre{\envstrat}{t} $
	\item[(E2)] \emph{Environment refusal}: $\forall t \in \tr^\sysstrat : t \notin \tr^\envstrat \land \pre{\sysstrat}{t} \subseteq \pl^\envstrat \Rightarrow \pre{\sysstrat}{t} \cap \plE^\envstrat \neq \emptyset$
	\item[(E3)] \emph{Progress}:\\ $\forall M \in \reach(\pNet^\envstrat) : \exists t_\sysstrat \in \tr^\sysstrat : \pre{\sysstrat}{t_\sysstrat} \subseteq M \Rightarrow \exists t_\envstrat \in \tr^\envstrat : \pre{\envstrat}{t_\envstrat} \subseteq M$
\end{itemize}

\emph{Explicit choice} requires each environment player to choose at most one of its outgoing transitions.
\emph{Environment refusal} enforces environment strategies to only remove transitions with at least one environment place in their preset. 
\emph{Progress} requires the environment strategy to enable at least one transition for each reachable marking as long as a transition is enabled in the system strategy. 
Environment strategies resolve the remaining conflicts of a Petri game:
\begin{theorem}
An environment strategy $\envstrat$ leads to a unique sequence of fired transitions up to reordering of independent transitions $(\forall p \in \pl^{\envstrat}: |\post{\envstrat}{p}|  \leq 1)$.\label{them:uniquerun}
\end{theorem}
\begin{proof}
	A system strategy $\sysstrat$ satisfies for all system places $p \in \pl^\sysstrat_S$ either the condition $ |\post{\sysstrat}{p}| \leq 1$ or the non-determinism in the choice of the successor transition is resolved by the environment strategy~$\envstrat$. 
	Since the environment strategy explicitly chooses at most one outgoing transition in each environment place, $\forall p \in \pl^{\envstrat}_S: |\post{\envstrat}{p}| \leq 1$ is satisfied. 
	For all environment places $ p \in \plE^{\envstrat}$, the condition $|\post{\envstrat}{p}| \leq 1$ is satisfied by the definition of environment strategies. 
	Since $\plS^{\envstrat} \cup \plE^{\envstrat} = \pl^{\envstrat}$ holds, $\pNet^{\envstrat}$ has a unique sequence of fired transitions up to reordering of independent transitions.
	\qed
\end{proof}

The requirements for environment strategies are similar to the ones for system strategies:
(E1) does not iterate over reachable markings in comparison to (S1) to require unique decisions by environment players, 
(E2) allows differentiation of transitions due to the unfolding in comparison to (S2), again, to enable unique decision, and 
(E3) is (S3) lifted directly to environment strategies.

$\envstrat \sqsubseteq_E \sysstrat$ denotes an environment strategy~$\envstrat$ as subprocess of a system strategy~$\sysstrat$ subject to (E1) to (E3). 
$\sysstrat \sqsubseteq_S \unf$ denotes a system strategy~$\sysstrat$ as a subprocess of the unfolding~$\unf$ subject to (S1) to (S3).
An environment strategy~$\envstrat$ is \emph{winning} (and a \emph{counterexample} to the system strategy $\sysstrat$ being winning) if it reaches a bad place. 
We define a system strategy to be \emph{winning} against all its environment strategies:
a system strategy $\sysstrat$ is winning if no bad places are reached for all environment strategies, i.e., $\forall \envstrat \sqsubseteq_E \sysstrat : \forall M \in \reach(\pNet^{\envstrat}): \lambda^{\envstrat}[M]   \cap \bad = \emptyset $.

Figure~\ref{fig:ProductionLineEnvStratgegy} shows a winning environment strategy for a system strategy of our running example with the bad place~$\bot$. 
By the initial decision for \textit{1_robot} by the environment strategy, the right side of the system strategy becomes unreachable.
The system chooses the transitions \textit{repair1} and \textit{repair2} in response to \textit{1_robot} by the environment strategy. 
By choosing \textit{1_robot}, the second robot should have ignored the product. 
The system strategy has to enable \textit{wrong_repair} to avoid a deadlock and the environment strategy agrees on firing it to reach the bad place.

\input{InputFigures/ProductionLineEnvStrategy.tex}

\subsection{True Concurrent Flow Semantics}
We define the \emph{true concurrent flow semantics} for Petri games by firing a maximal set of enabled, conflict-free transitions in every step. 
For the marking $M$ and the set of enabled, conflict-free transitions $T = \{t_1, \ldots, t_n\} $, the successor marking $\markingPG'$ is defined by $\markingPG\firable{T} \markingPG'$, where $\pre{\pNet}{t_1} \uplus  \ldots \uplus \pre{\pNet}{t_n} \subseteq \markingPG$ and $\markingPG' = (\markingPG \backslash (\pre{\pNet}{t_1}\uplus \ldots \uplus \pre{\pNet}{t_n})) \uplus \post{\pNet}{t_1} \uplus \ldots \uplus \post{\pNet}{t_n}$. 
The set of reachable markings according to the true concurrent flow semantics is defined by
\begin{eqnarray*}  
\reachtc(\pNet) & = &  \{\markingPG \subseteq \pl \mid \exists \text{ maximal } T_1, \ldots, T_n \subseteq \tr : \exists M_1,\ldots,M_n \subseteq \pl :  \\ 
&& \init\firable{T_1}M_1 \firable{T_2}  \ldots \firable{T_n} M_n=M\} \qquad \text{where }  |T_1|,\ldots, |T_n| > 0 
\end{eqnarray*}
We denote the set of reachable markings in the sequential flow semantics by $\reachseq(\pNet) = \reach(\pNet)$.
Firing all enabled transitions in the true concurrent flow semantics at once yields a unique sequence of markings and therefore a unique sequence of sets of fired transitions.
This brings us to the following theorem:

\begin{theorem}\label{th:equivalentexistence}
	There exists a winning system strategy of a Petri game under the sequential flow semantics iff there exists a winning system strategy of a Petri game under the true concurrent flow semantics.
\end{theorem}

\begin{proof}
	We show that 
	(1) $\exists \sysstrat \sqsubseteq_S \unf : \forall M \in \reachseq(\pNet^\sysstrat)  : \lambda^\sysstrat[M] \cap \bad = \emptyset \iff \exists \sysstrat \sqsubseteq_S \unf : \forall \envstrat \sqsubseteq_E \sysstrat : \forall M \in \reachseq(\pNet^{\envstrat}): \lambda^{\envstrat}[M] \cap \bad = \emptyset$ and that 
	(2) $\exists \sysstrat \sqsubseteq_S \unf : \forall \envstrat \sqsubseteq_E \sysstrat : \forall M \in \reachseq(\pNet^{\envstrat}): \lambda^{\envstrat}[M] \cap \bad = \emptyset \iff \exists \sysstrat \sqsubseteq_S \unf : \forall \envstrat \sqsubseteq_E \sysstrat : \forall M \in \reachtc(\pNet^{\envstrat}): \lambda^{\envstrat}[M] \cap \bad = \emptyset$. 
	Since (1) is based on the sequential flow, every sequence of markings in $\reachseq(\pNet^\sysstrat)$ can be produced with an environment strategy choosing exactly the transitions of the sequence and vice versa. 
	For (2), we show that the environment wins on the same nets by reaching a bad place: 
	either $\exists \envstrat \sqsubseteq_E \sysstrat : \exists \markingPG \in \reachseq(\pNet^{\envstrat}): \lambda^\envstrat [M] \cap \bad \neq \emptyset$ holds or not. 
	As each environment strategy results in a unique sequence of fired transitions (up to reordering of independent transitions), the sets of reachable places in the reachable markings $\reachseq{(\pNet^\envstrat)}$ and $\reachtc{(\pNet^\envstrat)}$ are the same.\qed
\end{proof}

\section{True Concurrent Encoding of Bounded Synthesis}
\label{sec:qbf}

We show how the requirements (E1) to (E3) on environment strategies and the true concurrent flow semantics can be encoded as QBF. 
We introduce \emph{stalling} of transitions to let environment players find non-determinism in a system strategy.
Furthermore, we present how the true concurrent flow semantics can be used to detect loops earlier in the encoding of bounded synthesis for Petri games.

\subsection{Stalling of Transitions to Find Non-Determinism}

To use the true concurrent flow semantics in bounded synthesis for Petri games, we ensure that all possible system strategies fulfill the assumptions (S1) to (S3) and do not reach any bad place.
The determinism requirement can be violated when the sequential flow encoding is simply replaced by the true concurrent flow encoding as markings may be skipped by firing transitions as early as possible. 

Figure~\ref{fig:NonDet} shows a Petri game without bad places. 
It is not winning for the system, as \textit{t4} and \textit{t6} have to be enabled (deadlock-avoidance) and there is a marking where both transitions are enabled (non-determinism). 
This contradicts \emph{determinism} (S1) but in the true concurrent flow semantics, \textit{t4} will always be fired before \textit{t6} such that the marking with non-determinism of the system is never reached. 
To check the requirements for system strategies in the true concurrent encoding, environment players can \emph{stall} transitions with at least one system place in their preset globally to catch up with the system.
The requirement \emph{determinism} (S1) can only be violated at system places.
In Fig.~\ref{fig:NonDet}, the environment strategy needs to stall the firing of $t4$ until $t5$ is fired to prove that a potential system strategy enabling both transitions is non-deterministic. 

\subsection{Encoding True Concurrency as QBF}\label{sec:trueConEncoding}

We extend the sequential encoding of bounded synthesis for Petri games \cite{Finkbeiner15,Jesko} to environment strategies with stalling and the true concurrent flow semantics. 
The strategy of the environment is translated into additional universally quantified variables. 
The QBF-formula is $\exists \mathscr{S}^b : \forall \mathscr{M}_n : \forall \mathscr{E}^b : \phi_n$ with $\mathscr{E}^b$ as the union of variables for each environment choice in the firing of transitions and variables for transitions with at least one system place in their preset to stall their progress.
This encoding preserves the requirement of \emph{environment refusal} (E2):
\begin{eqnarray*}
	\mathscr{E}^b & \defEQ & \{(p, t, i) \mid p \in \plE^\bd \land t \in p^\bullet \land 1 \leq i < n\}  \cup \{ (t) \mid t \in \tr^\bd \land {^\bullet t} \cap \plS^\bd \neq \emptyset \}
 \end{eqnarray*}
 
Bounded unfoldings may contain loops. 
The variables for the environment are different for every simulation point, such that decisions of revisited environment places do not depend on previous visits. 
By contrast, a global decision independent of the simulation points suffices for stalling. 
The case when variable $(t)$ is set to false results in the stalling of transition~$t$. 
In the following, not mentioned formulas are as they are in the sequential encoding. 
We apply the requirement \emph{explicit choice} (E1) of the environment strategy to $\phi_n$ and encode it in $\mathit{choice}$:
\begingroup
\allowdisplaybreaks
\begin{eqnarray*}
	\phi_n & \defEQ &  \mathit{choice} \implies 
	\bigwedge_{1 \leq i < n} \Bigg( \mathit{seq}_i \implies \mathit{win}_i \Bigg) \wedge \big( \mathit{seq}_n \implies \mathit{win}_n \wedge \mathit{loop} \big) \\
	\mathit{choice} & \defEQ & \bigwedge_{p \in \places_E^b, 1 \leq i < n} \Bigg(\bigvee_{t \in p^\bullet}\Big( (p, t,i) \wedge \bigwedge_{t' \in p^\bullet \setminus \{t\}}\neg(p, t',i) \Big) \Bigg)\\
	\mathit{seq}_i & \defEQ & \mathit{initial} \wedge \mathit{tcflow}_1 \wedge \mathit{tcflow}_2 \wedge \dots \wedge \mathit{tcflow}_{i-1}
\end{eqnarray*}
\endgroup

Each environment place has to choose exactly one outgoing transition which results in the firing of at most one outgoing transitions per environment place, because the other places in the preset of the transition also have to decide for the transition.
This encoding furthermore ensures \emph{progress} (E3).
We substitute the sequential flow $\mathit{seqflow}_i$ by the true concurrent flow $\mathit{tcflow}_i$, which enforces the firing of all enabled and not stalled transitions and maintains all other tokens.
\begingroup
\allowdisplaybreaks
\begin{eqnarray*}
	\mathit{tcflow}_i & \defEQ & \mathit{fireenabled}_{i} \wedge \mathit{updateplaces}_{i} \\
	\mathit{fireenabled}_{i} & \defEQ & \bigwedge_{t \in \tr^\bd} \Bigg( \mathit{enabled}_{i,t} \implies \bigwedge_{p\in {^\bullet t} \setminus t^\bullet}\neg(p,i+1) \wedge \bigwedge_{p\in t^\bullet} (p,i+1) \Bigg) \\
	\mathit{updateplaces}_{i} & \defEQ & \bigwedge_{p \in \pl^\bd}\Bigg( \bigwedge_{t \in {^\bullet p} \cup p^\bullet} \neg \mathit{enabled}_{i,t} \implies \big((p,i) \iff (p,i+1)\big)\Bigg)\\
	\mathit{enabled}_{i,t}&\defEQ & \bigwedge_{p \in {^\bullet t}}(p,i) \wedge \bigwedge_{p \in \plS^\bd\cap{^\bullet t}} (p,\lambda^\bd (t)) \wedge \bigwedge_{p \in \plE^\bd\cap{^\bullet t}} (p, t,i) \wedge (t) 
\end{eqnarray*}
\endgroup

$\mathit{enabled}_{i,t}$ requires tokens in all places in the preset of the transition, both the system and the environment strategy to allow the transition for corresponding places in the preset of the transition, and that stalling allows the transition. 

$\mathit{win}_i$ remains unchanged.
Therefore, environment strategies and stalling only affect the flow of tokens but not the check that reached markings are winning.

\subsection{Shorter Loops via Strongly Connected Components}

Environment strategies allow us to define the true concurrent flow semantics which allows us to detect loops earlier by searching for them in \emph{strongly connected components} (SCCs)~\cite{jensen2013coloured}.
The definition of SCCs can be directly lifted to Petri games by including an additional set with all places that are not in any other SCC.
With SCCs, we find loops in independent parts of the Petri game as early as possible. 
We encode that a loop no longer only occurs at the repetition of a global marking but also when all $\mathit{SCCs} \subseteq 2^{\pl^\bd}$ repeat their marking, respectively:
\begin{eqnarray*}
	\mathit{loop} & \defEQ & \bigwedge_{\mathit{scc} \in \mathit{SCCs}}\Bigg(\bigvee_{1\leq i_1 < i_2 \leq n} \Big(\bigwedge_{p \in \mathit{scc}} \big((p,i_1) \iff (p,i_2) \big) \Big) \Bigg)
\end{eqnarray*}

%% file: InputFigures/ProductionLineEnvStrategy.tex
% !TeX root = ../Paper.tex
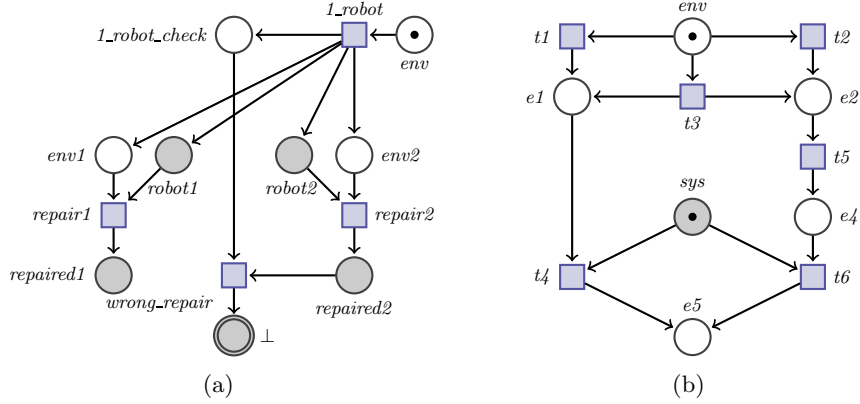
\begin{figure}[t]
	%\captionsetup[subfigure]{font=footnotesize}
	\centering
	
		\subcaptionbox{\label{fig:ProductionLineEnvStratgegy}}[.5\linewidth]{
	\begin{tikzpicture}[scale=0.8, every node/.style={transform shape}]
	%Gobal Flow places
	\node [envplace] (env) [tokens = 1, label=below:\textit{env}]{};
	%\node [sysplace] (sys1) [tokens = 1, left of = env, label = above:\textit{robot1}] {};
	%\node [sysplace] (sys2) [tokens = 1, right of = env, label = above:\textit{robot2}] {};
	
	%\node [sysplace] (sysstart) [below of = env, below of = env,label=below right:\textit{sysstart}]{};
	
	\node [envplace] (e1) [left of = env, left of = env, left of = env, left of=env, left of = env, below of = env, below of = env, label=left:\textit{env1}]{};
	\node [sysplace] (s1) [right of = e1, label=below:\textit{robot1}]{};
	\node [sysplace] (s2) [right of = s1, right of = s1, label=below:\textit{robot2~~}]{};
	\node [envplace] (e2) [right of = s2, label=right:\textit{env2}]{};

	%\node [sysplace] (done11) [below of = s1, below of = s1, label = above:\textit{ignored1}]{};
	\node [sysplace] (done12) [below of = s1, below of = s1,  left of = s1, label =left:\textit{repaired1}]{};
	
	%\node [sysplace] (done22) [below of = s2, below of = s2,  right of = s2, label =above:\textit{repaired2}]{};

	%Bad behavior check places
	\node [envplace] (tcheck1) [left of= env, left of= env, left of = env, label=left:\textit{1\_robot\_check}]{};

	%Syscheck for right decision of systems:
	\node [sysplace] (done22) [below of = s2, below of = s2,  right of = s2, label = below:\textit{repaired2}]{};
	
	%Flow transitions
	%\node [transition] (sysinit) [below of = sysstart,  label=below :\textit{init}]{};
	\node [transition](t1) [left of= env, label=above:\textit{1\_robot}]{};

	%\node [transition] (tdec11) [below of = s1, label =  right:\textit{ignore1}]{};
	\node [transition] (tdec12) [below of = s1, left of= s1, label = left:\textit{repair1}]{};
	\node [transition] (tdec22) [below of = s2, right of= s2, label = right:\textit{repair2}]{};
	
	%\node [transition] (tdec22) [below of = s2, right of= s2, label = right:\textit{repair2}]{};
	\node [transition] (tbad5) [ right of= done12, right of= done12, label = below left:\textit{wrong\_repair}]{};
	
	\node [sysplace, bad] (bad) [below of = tbad5, label = right:$\bot$]{};

	\path[->, thick] (t1) 
	%edge [pre]  (sys1)
	%edge [pre]  (sys2)
	edge [pre]  (env)
	edge [post] (tcheck1)
	edge [post] (s2)
	edge [post] (s1)
	edge [post] (e2)
	edge [post] (e1);
	
	\path[->, thick] (tdec22) 
	edge [pre]  (e2)
	edge [pre] (s2)
	edge [post] (done22);
	
	\path[->, thick] (tdec12) 
	edge [pre]  (e1)
	edge [pre] (s1)
	edge [post] (done12);

	\path [->, thick] (tbad5)
	edge [pre] (done22)
	edge [pre] (tcheck1)
	edge [post] (bad);

	\end{tikzpicture}}
	\qquad
	\subcaptionbox{\label{fig:NonDet}}[.4\linewidth]{		
	\begin{tikzpicture}[scale=0.8, every node/.style={transform shape}]
		%nodes
		\node [envplace] (env) [tokens = 1, label=above:\textit{env}]{};
		%transitions
		\node [transition](t1) [left of= env, left of = env, label= left:\textit{t1}]{};
		\node [transition](t2) [right of= env, right of=env, label= right:\textit{t2}]{};
		
		\node [envplace] (e1) [ below of=t1, label=left:\textit{e1}]{};

		\node [transition](t4) [below of=e1, below of= e1, below of = e1,  label= left:\textit{t4}]{};
		\node [envplace](e2) [below of=t2,   label=right:\textit{e2}]{};
		\node [transition](t3) [below of=env, label=below :\textit{t3}]{};
		\node [transition] (t5) [below of =e2, label = right: \textit{t5}]{};
		\node[sysplace] (sys) [tokens=1, below of = t3, below of=t3, label=above:\textit{sys}]{};
		
		\node [envplace](e4) [below of=t5,   label=right:\textit{e4}]{};
		
		\node [transition] (t6) [below of =t5, below of=t5, label = right: \textit{t6}]{};
		
		\node [envplace] (e5) [below of = sys, below of = sys, label = above:\textit{e5}]{};
		
		\path[->, thick] (t1) 
		edge [pre]  (env)
		edge [post] (e1);
 		
 		\path[->, thick] (t2) 
 		edge [pre]  (env)
 		edge [post] (e2);
 		
 		\path[->, thick] (t3) 
 		edge [pre]  (env)
 		edge [post] (e1)
 		edge [post] (e2);
 		
 		\path[->, thick] (t5) 
 		edge [pre]  (e2)
 		edge [post] (e4);
 		
		 \path[->, thick] (t4) 
		 edge [pre] (sys)
		 edge [pre]  (e1)
		 edge [post] (e5);	
		 
		 \path[->, thick] (t6) 
		 edge [pre] (sys)
		 edge [pre]  (e4)
		 edge [post] (e5);	
		 
 		\end{tikzpicture}}
	\caption{Two strategies are depicted for the Petri game specifying a production line from \refFig{ProductionLine}: a winning environment strategy for a system strategy~(a) and a winning system strategy with more than one outgoing transition at place \emph{sys}~(b).
	}
\end{figure}

%% file: InputFiles/ImplementationAndResults.tex
We compare the sequential encoding~\cite{Finkbeiner15} with our new true concurrent encoding from Section~\ref{sec:qbf} on five benchmark families.
At first, we describe the asynchronous and distributed nature of these benchmark families stemming from alarm systems, routing, robotics, and communication protocols. 
Afterwards, we outline the technical details of our comparison framework and state our observations and explanations concerning the observed times for finding winning strategies.

\subsection{Benchmark Families}
Table~\ref{fig:table} refers to the following scalable benchmark families where Collision Avoidance, Disjoint Routing, and Production Line are new benchmark families: 
\begin{itemize}
	\item \textbf{AS}: \emph{Alarm System} \cite{Jesko}. \emph{Parameters:} $m$ locations.
	There are $m$ secured locations and a burglar can intrude one of them. 
	The local alarm system of each location can communicate with all other local alarm systems. 
	The local alarm systems should indicate the position of an intrusion and should not issue unsubstantiated warnings of an intrusion.
	\item \textbf{CA:} \emph{Collision Avoidance}. \emph{Parameters:} $m$ robots. A subset of $m$ robots is initialized to drive on individual paths of increasing length with several goal states. 
	They should avoid collisions and drive forever on the chosen route.
	\item \textbf{DR:} \emph{Disjoint Routing}. \emph{Parameters:} $m$ packets. 
	In a software-defined network, $m$ packets should be routed disjointly between an ingress and an egress switch where the network allows $m$ disjoint paths between the two switches.
	\item \textbf{PL:} \emph{Production Line}. \emph{Parameters:} $m$ robots.
	The $m$ independent robots are able to repair or ignore $m$ features of a product. 
	Depending on the product, some features need to be repaired while others must not be repaired.
	\item \textbf{DW:} \emph{Document Workflow} \cite{ADAM}. \emph{Parameters:} $m$ workers.
	A document circulates between $m$ workers with the environment choosing the first worker.
	It is required that all workers unanimously endorse or reject the document.
\end{itemize}

\begin{table}[t]
	\centering
	\caption{Benchmarking results on our Petri game \emph{benchmark families} for increasing \emph{parameters}. For the \emph{sequential} and the \emph{true concurrent} encoding, the needed model checking \textit{iterations} with accumulated \emph{runtime in seconds} are reported.
	}
	\label{fig:table}
 	\begin{tabular}{c|c|c|c|c|c}
		\hline
		\hline
		&& \multicolumn{2}{|c}{\textit{Sequential}} & \multicolumn{2}{|c}{\textit{True Concurrent}}\\     
		\textit{Ben.} & \textit{Par.} & \textit{Iter.} & \textit{Runtime in sec.} & \textit{Iter.} & \textit{Runtime in sec.} \\
		\hline
		AS & 2 & 7 & 13.26 & \textbf{6} & \textbf{11.15}\\
		& 3 & - & timeout & - & timeout \\ 
		%old: 
		%AS & 2 & 7 & 75.55 & \textbf{6} & \textbf{64.58}\\
		%& 3 & - & timeout & - & timeout \\ 
		\hline
		CA & 2 & 8 & 7.27 & \textbf{5} & \textbf{6.25}\\
		& 3 & - & timeout & \textbf{6} & \textbf{14.21}\\
		& 4 & - & timeout & \textbf{7} & \textbf{346.23}\\
		& 5 & - & timeout& -& timeout\\
        %old:
        %CA & 2 & 8 & 40.63 & \textbf{5} & \textbf{22.48}\\
        %& 3 & - & timeout & \textbf{6} & \textbf{36.37}\\
        %& 4 & - & timeout & \textbf{7} & \textbf{369.33}\\
        %& 5 & - & timeout& -& timeout\\
        \hline
        DR & 2 & 8 & 6.16 & \textbf{7}& \textbf{6.05}\\
        & 3 & 11 & 11.03 & \textbf{9}& \textbf{10.07}\\
        & 4 & 14 & 69.50 & \textbf{11} & \textbf{65.31}\\
		& 5 & - & timeout& -& timeout\\
        %old:
        %DR & 2 & 8 & 39.22 & \textbf{5}& \textbf{22.32}\\
        %& 3 & 11 & 63.07 & \textbf{5}& \textbf{24.12}\\
        %& 4 & 14 & 153.37 & \textbf{5} & \textbf{34.88}\\
		%& 5 & - & timeout & \textbf{5}& \textbf{642.26}\\
        %& 6 & - & timeout& -& timeout\\
		\hline
		PL & 1 & \textbf{4} & \textbf{5.59} & \textbf{4} & \textbf{5.59}\\
		& 2 & 5 & 6.08 & \textbf{4} & \textbf{5.85}\\ 
		%& $\ldots$  & $\ldots$ & \ldots & \ldots & \ldots \\ 
		& 3 & 6 & 8.51 & \textbf{4} & \textbf{6.95}\\
		& 4 & 7 & 20.99 & \textbf{4} & \textbf{12.54}\\
		& 5 & 8 & 87.33 & \textbf{4} & \textbf{41.95} \\
        & 6 & - & timeout & \textbf{4}& \textbf{742.36}\\
        & 7 & - & timeout & - & timeout \\
        %old: 
		%PL & 1 & \textbf{4} & \textbf{16.57} & \textbf{4} & 16.58\\
		%& 2 & 5 & 22.75 & \textbf{4} & \textbf{16.88}\\ 
		%& 3 & 6 & 31.33 & \textbf{4} & \textbf{18.18}\\
		%& 4 & 7 & 52.85 & \textbf{4} & \textbf{23.86}\\
		%& 5 & 8 & 174.46 & \textbf{4} & \textbf{67.01} \\
        %& 6 & - & timeout & \textbf{4}& \textbf{739.57}\\
        %& 7 & - & timeout & - & timeout \\ 
		\hline
		DW & 1 & 8 & 5.90 & \textbf{7} & \textbf{5.79} \\
		& 2 &  10 & 6.58 & \textbf{9} & \textbf{6.44} \\  
		& 3 &  12 & 7.90 & \textbf{11} & \textbf{7.80} \\  
		& 4 &  14 & 11.45 & \textbf{13} & \textbf{11.22} \\
		& 5 &  16 & \textbf{16.59} & \textbf{15} & 19.82 \\
		& $\ldots$  & $\ldots$ & \ldots & \ldots & \ldots \\ 
		%&  6&  18 & \textbf{26.71} & \textbf{17} & 31.79 \\
		%&  7&  20 & \textbf{49.69} & \textbf{19} & 56.38\\
		%&  8&  22 & \textbf{90.41} & \textbf{21} & 132.73 \\
		%&  9&  24 & \textbf{239.47} & \textbf{23} & 327.68 \\
		& 10&  26 & \textbf{716.61} & \textbf{25} & 823.94 \\
		& 11 & \textbf{28} & \textbf{1304.14} & - & timeout \\
		& 12 & - & timeout & - & timeout \\
		%old:
		%DW & 1 & 8 & 33.48 & \textbf{7} & \textbf{31.01} \\
		%& 2 &  10 & 46.19 & \textbf{9} & \textbf{40.02} \\  
		%& 3 &  12 & 57.52 & \textbf{11} & \textbf{50.56} \\ 
		%& $\ldots$  & $\ldots$ & \ldots & \ldots & \ldots \\  
		%& 4 &  14 & \textbf{69.80} & \textbf{13} & 70.32 \\
		%& 5 &  16 & 90.67 & \textbf{15} & \textbf{90.40} \\
		%&  6&  18 & \textbf{131.65} & \textbf{17} & 139.44 \\
		%&  7&  20 & \textbf{179.47} & \textbf{19} & 201.06\\
		%&  8&  22 & \textbf{272.04} & \textbf{21} & 325.55 \\
		%&  9&  24 & \textbf{506.46} & \textbf{23} & 671.04 \\
		%&  10&  26 & 1038.69 & \textbf{25} & \textbf{1002.45} \\
		%& 11 & - & timeout & - & timeout \\ 
        \hline
        \hline
	\end{tabular}
\end{table}

\subsection{Comparison Framework}

As both the sequential and the true concurrent encoding result in a 2-QBF not in conjunctive normal form, we use the QBF solver \textsc{QuAbS}~\cite{DBLP:conf/sat/Tentrup16,DBLP:journals/corr/Tentrup16}.
The results from Table~\ref{fig:table} were obtained on an Intel i7-2700K CPU with 3.50~GHz and 32~GB RAM and are the average over five runs.
For each benchmark family (column \textit{Ben.}), we report on the attempted parameters of the benchmark (\textit{Par.}), the necessary model checking iterations (\textit{Iter.}) of bounded synthesis, and on the runtime for finding a winning system strategy. 
A timeout of 30 minutes is used.
We prepared an artifact to replicate our experimental results~\cite{HM19}.

\subsection{Observation}

The true concurrent encoding shows considerable improvements over the sequential encoding on the presented benchmark set: It solves more instances and has mostly faster solving times as shown in Table~\ref{fig:table}.
The improvements are based on fewer model checking iterations of the bounded synthesis algorithm witnessed by the \textit{Iter.} column. 
The lower iteration count and runtime are indicated in bold.

We can make the following observations concerning the specific benchmark families:
The complex communication structure of Alarm System prevents larger examples to be synthesized because the alarm system observing the intrusion has to broadcast the information to all other alarm systems.
Similarly, Collision Avoidance has a complex pairwise communication structure which can be better synthesized by the true concurrent encoding.
The simpler communication structure of Production Line allows constant bounds for the true concurrent encoding compared to linearly increasing bounds for the sequential encoding.
The communication structure of Disjoint Routing lays between complex and simple such that the true concurrent encoding enables a smaller linear increase in the bound.  
The true concurrent encoding therefore can solve larger examples even though the bounded unfolding grows with the number of considered players for both encodings.
The possibilities for communication of information are less open in the benchmark families DR and PL whereas they are completely open in the benchmark family AS and CA.
In Document Workflow, the communication structure is fixed to a specific pairwise ring between neighboring clerks.
However, this prevents almost all true concurrency between them.
The difference in bound of one is caused by the concurrent test that all workers have seen the document and that the decisions of workers have been unanimously.

%% file: Paper.bbl
\begin{thebibliography}{10}
\providecommand{\url}[1]{\texttt{#1}}
\providecommand{\urlprefix}{URL }
\providecommand{\doi}[1]{https://doi.org/#1}

\bibitem{CONCUR19}
Beutner, R., Finkbeiner, B., Hecking-Harbusch, J.: Translating asynchronous
  games for distributed synthesis. In: Proceedings of {CONCUR}. pp. 26:1--26:16
  (2019)

\bibitem{DBLP:conf/cav/BohyBFJR12}
Bohy, A., Bruy{\`{e}}re, V., Filiot, E., Jin, N., Raskin, J.: Acacia+, a tool
  for {LTL} synthesis. In: Proceedings of {CAV}. pp. 652--657 (2012)

\bibitem{DBLP:journals/tcs/BonetHKTV14}
Bonet, B., Haslum, P., Khomenko, V., Thi{\'{e}}baux, S., Vogler, W.: Recent
  advances in unfolding technique. Theor. Comput. Sci.  \textbf{551},  84--101
  (2014)

\bibitem{DBLP:conf/tacas/Ehlers11}
Ehlers, R.: Unbeast: Symbolic bounded synthesis. In: Proceedings of {TACAS}.
  pp. 272--275 (2011)

\bibitem{DBLP:series/eatcs/EsparzaH08}
Esparza, J., Heljanko, K.: Unfoldings -- {A} Partial-Order Approach to Model
  Checking. Springer (2008)

\bibitem{DBLP:conf/tacas/FaymonvilleFRT17}
Faymonville, P., Finkbeiner, B., Rabe, M.N., Tentrup, L.: Encodings of bounded
  synthesis. In: Proceedings of {TACAS}. pp. 354--370 (2017)

\bibitem{Finkbeiner15}
Finkbeiner, B.: Bounded synthesis for petri games. In: Correct System Design.
  pp. 223--237 (2015)

\bibitem{Jesko}
Finkbeiner, B., Gieseking, M., Hecking{-}Harbusch, J., Olderog, E.: Symbolic
  vs. bounded synthesis for petri games. In: Proceedings of SYNT. pp. 23--43
  (2017)

\bibitem{ADAM}
Finkbeiner, B., Gieseking, M., Olderog, E.: Adam: Causality-based synthesis of
  distributed systems. In: Proceedings of {CAV}. pp. 433--439 (2015)

\bibitem{DBLP:conf/fsttcs/FinkbeinerG17}
Finkbeiner, B., G{\"{o}}lz, P.: Synthesis in distributed environments. In:
  Proceedings of {FSTTCS}. pp. 28:1--28:14 (2017)

\bibitem{petrigamessynthesis}
Finkbeiner, B., Olderog, E.: Petri games: Synthesis of distributed systems with
  causal memory. Inf. Comput.  \textbf{253},  181--203 (2017)

\bibitem{DBLP:conf/lics/FinkbeinerS05}
Finkbeiner, B., Schewe, S.: Uniform distributed synthesis. In: Proceedings of
  {LICS}. pp. 321--330 (2005)

\bibitem{boundedsynthesis}
Finkbeiner, B., Schewe, S.: Bounded synthesis. {STTT}  \textbf{15}(5-6),
  519--539 (2013)

\bibitem{DBLP:conf/popl/FlanaganG05}
Flanagan, C., Godefroid, P.: Dynamic partial-order reduction for model checking
  software. In: Proceedings of {POPL}. pp. 110--121 (2005)

\bibitem{DBLP:conf/fsttcs/GastinLZ04}
Gastin, P., Lerman, B., Zeitoun, M.: Distributed games with causal memory are
  decidable for series-parallel systems. In: Proceedings of {FSTTCS}. pp.
  275--286 (2004)

\bibitem{DBLP:conf/icalp/GenestGMW13}
Genest, B., Gimbert, H., Muscholl, A., Walukiewicz, I.: Asynchronous games over
  tree architectures. In: Proceedings of {ICALP}. pp. 275--286 (2013)

\bibitem{DBLP:conf/fsttcs/Gimbert17}
Gimbert, H.: On the control of asynchronous automata. In: Proceedings of
  {FSTTCS}. pp. 30:1--30:15 (2017)

\bibitem{HM19}
Hecking-Harbusch, J., Metzger, N.O.: {BoundedAdam – Efficient Trace Encodings
  for Bounded Synthesis of Petri Games}  (2019).
  \doi{10.6084/m9.figshare.8313215}

\bibitem{DBLP:journals/corr/Tentrup16}
Hecking{-}Harbusch, J., Tentrup, L.: Solving {QBF} by abstraction. In:
  Proceedings of GandALF. pp. 88--102 (2018)

\bibitem{DBLP:journals/fuin/Heljanko99}
Heljanko, K.: Using logic programs with stable model semantics to solve
  deadlock and reachability problems for 1-safe petri nets. Fundam. Inform.
  \textbf{37}(3),  247--268 (1999)

\bibitem{DBLP:phd/basesearch/Heljanko02}
Heljanko, K.: Combining symbolic and partial order methods for model checking
  1-safe Petri nets. Ph.D. thesis, Aalto University, Helsinki, Finland (2002)

\bibitem{jensen2013coloured}
Jensen, K.: Coloured Petri nets: basic concepts, analysis methods and practical
  use, vol.~2. Springer Science \& Business Media (2013)

\bibitem{DBLP:conf/cav/JobstmannGWB07}
Jobstmann, B., Galler, S.J., Weiglhofer, M., Bloem, R.: Anzu: {A} tool for
  property synthesis. In: Proceedings of {CAV}. pp. 258--262 (2007)

\bibitem{DBLP:journals/acta/KhomenkoKV03}
Khomenko, V., Koutny, M., Vogler, W.: Canonical prefixes of petri net
  unfoldings. Acta Inf.  \textbf{40}(2),  95--118 (2003)

\bibitem{DBLP:conf/concur/MadhusudanT02}
Madhusudan, P., Thiagarajan, P.S.: A decidable class of asynchronous
  distributed controllers. In: Proceedings of {CONCUR}. pp. 145--160 (2002)

\bibitem{DBLP:conf/fsttcs/MadhusudanTY05}
Madhusudan, P., Thiagarajan, P.S., Yang, S.: The {MSO} theory of connectedly
  communicating processes. In: Proceedings of {FSTTCS}. pp. 201--212 (2005)

\bibitem{DBLP:conf/wotug/MeulenP09}
Meulen, J.V., Pecheur, C.: Combining partial order reduction with bounded model
  checking. In: Proceedings of CPA. pp. 29--48 (2009)

\bibitem{DBLP:conf/icalp/Muscholl15}
Muscholl, A.: Automated synthesis of distributed controllers. In: Proceedings
  of {ICALP}. pp. 11--27 (2015)

\bibitem{DBLP:conf/fsttcs/MuschollW14}
Muscholl, A., Walukiewicz, I.: Distributed synthesis for acyclic architectures.
  In: Proceedings of {FSTTCS}. pp. 639--651 (2014)

\bibitem{DBLP:conf/icalp/PnueliR89}
Pnueli, A., Rosner, R.: On the synthesis of an asynchronous reactive module.
  In: Proceedings of {ICALP}. pp. 652--671 (1989)

\bibitem{DBLP:conf/focs/PnueliR90}
Pnueli, A., Rosner, R.: Distributed reactive systems are hard to synthesize.
  In: Proceedings of FOCS. pp. 746--757 (1990)

\bibitem{DBLP:books/sp/Reisig85a}
Reisig, W.: Petri Nets: An Introduction. Springer (1985)

\bibitem{DBLP:conf/sat/Tentrup16}
Tentrup, L.: Non-prenex {QBF} solving using abstraction. In: Proceedings of
  {SAT}. pp. 393--401 (2016)

\end{thebibliography}
